\documentclass{article}

\usepackage[all]{xy}
\usepackage{natbib}
\usepackage[margin=1in]{geometry}
\usepackage{fullpage}
\usepackage{graphicx}
\usepackage{hyperref}
\usepackage{tikz}
\usepackage{enumitem}
\usepackage{xcolor}
\usepackage[utf8]{inputenc}
\usepackage{booktabs}
\hypersetup{
   breaklinks,%
   colorlinks=true,%
   linkcolor=black,%
   urlcolor=black,%
   citecolor=[rgb]{0,0,0.45}
}

\def\bar{\overline}

\def\script#1{\mathcal{#1}}

\def\E{\mathbf{E}}

\newcommand{\poly}{\mathop{\mathrm{poly}}}%

\def\sA{\script{A}}

\def\sC{\script{C}}
\def\sD{\script{D}}
\def\sF{\script{F}}

\def\sI{\script{I}}
\def\sJ{\script{J}}

\def\sN{\script{N}}

\def\sR{\script{R}}

\def\sX{\script{X}}

\def\opt{\textsc{OPT}}

\usepackage{amsmath}
\usepackage{amssymb}
\usepackage{amsthm}
\newtheorem{theorem}{Theorem}[section]

\newtheorem{lemma}[theorem]{Lemma}

\newtheorem{corollary}[theorem]{Corollary}
\newtheorem{claim}[theorem]{Claim} 

\newtheorem{definition}[theorem]{Definition}

\newtheorem{remark}[theorem]{Remark}%

\newcommand{\loc}{\operatorname{loc}}

\usepackage{algorithmic}            
\usepackage{algorithm}

\begin{document}

\title{Fair Representation Clustering with Several Protected Classes}

\author{Zhen Dai\thanks{University of Chicago. Supported by grants from NSF and DARPA. Email: \texttt{zhen9@uchicago.edu}} \and Yury Makarychev\thanks{Toyota Technological Institute at Chicago (TTIC). Supported by NSF awards CCF-1718820, CCF-1955173, and CCF-1934843. Email: \texttt{yury@ttic.edu}} \and Ali Vakilian\thanks{Toyota Technological Institute at Chicago (TTIC). Supported by NSF award CCF-1934843. Email: \texttt{vakilian@ttic.edu}}}
\date{}
\maketitle
\begin{abstract}
We study the problem of fair $k$-median where each cluster is required to have a fair representation of individuals from different groups. In the fair representation $k$-median problem, we are given a set of points $X$ in a metric space. Each point $x\in X$ belongs to one of $\ell$ groups. Further, we are given fair representation parameters $\alpha_j$ and $\beta_j$ for each group $j\in [\ell]$. We say that a $k$-clustering $C_1, \cdots, C_k$ fairly represents all groups if the number of points from group $j$ in cluster $C_i$ is between $\alpha_j |C_i|$ and $\beta_j |C_i|$ for every $j\in[\ell]$ and $i\in [k]$. The goal is to find a set $\sC$ of $k$ centers and an assignment $\phi: X\rightarrow \sC$ such that the clustering defined by $(\sC, \phi)$ fairly represents all groups and minimizes the $\ell_1$-objective $\sum_{x\in X} d(x, \phi(x))$.

We present an $O(\log k)$-approximation algorithm that runs in time $n^{O(\ell)}$. Note that the known algorithms for the problem either (i) violate the fairness constraints by an additive term or (ii) run in time that is exponential in both $k$ and $\ell$. We also consider an important special case of the problem where $\alpha_j = \beta_j = \frac{f_j}{f}$ and $f_j, f \in \mathbb{N}$ for all $j\in [\ell]$. For this special case, we present an $O(\log k)$-approximation algorithm that runs in $(kf)^{O(\ell)}\log n + \poly(n)$ time.  
\end{abstract}

\section{Introduction}
Algorithmic decision making is widely used for high-stake decisions like college admissions~\citep{marcinkowski2020implications} and criminal justice~\citep{chouldechova2017fair,kleinberg2018human}. While automated decision-making processes are often very efficient, there are serious concerns about their fairness.  Consequently, in recent years, there has been an extensive line of research on fairness of algorithms and machine learning approaches~\citep{chouldechova2018frontiers, kearns2019ethical,kleinberg2017inherent}.    

In this paper, we study the {\em ``fair representation''} clustering problem proposed in the seminal work of~\citet{chierichetti2017fair}.  
The notion, which is motivated by the concept of {\em disparate impact}~\citep{feldman2015certifying}, requires that each protected class has an approximately equal representation in each cluster. In many scenarios, a different set of benefits are associated with each cluster of points output by the algorithm. Then, it is desirable that different groups of individuals (e.g., men or women) receive the benefits associated with each of the clusters (e.g., mortgage options) in similar proportions. Further, clustering is often used for feature engineering. In this case, we need to ensure that the generated features are fair; that is, they neither introduce new nor amplify existing biases in the data set.        
Now, we formally define the notion of representation fairness for clustering.
\begin{definition}[fair representation clustering]\label{def:fairness}
Given a set of points $X$ that come from $\ell$ different groups $X_1, \dots, X_\ell$, a $k$-clustering $C_1, \cdots, C_k$ of $X$ is {\em fair} with respect to the fairness requirement specified by $\{\alpha_j, \beta_j\}_{j\in \ell}$ if
\begin{align}\label{eq:fair-cond}
    \forall i\in [k], j\in [\ell],\quad \alpha_j |C_i| \le |C_i \cap X_j| \le \beta_j |C_i| 
\end{align}

In fair $k$-median with fairness requirement $\{\alpha_j, \beta_j\}_j$, the goal is to find $k$ clusters $C_1, \dots, C_k$ and $k$ centers, $c_1,\dots, c_k$ (one center for each cluster) so that the clustering $C_1,\dots, C_k$ is fair with respect to the fairness requirement and the $\ell_1$-objective $\sum_{i=1}^k\sum_{x\in C_i} d(x, c_i)$ is minimized. We will say that points in $C_i$ are assigned to center $c_i$. We let $\phi$ be the assignment function that maps each point $u$ to the center $u$ is assigned to. To specify a solution, it is sufficient to provide the set of centers and $\phi$.
\end{definition}
\citet{bera2019fair} and~\citet{bercea2019cost} independently introduced this notion of fairness, which generalizes the notions studied by~\cite{chierichetti2017fair,schmidt2019fair, backurs2019scalable, ahmadian2019clustering}. \citeauthor{bercea2019cost} presented a constant factor approximation algorithm for the fair representation clustering with the general $\ell_p$-objective. 
However, their algorithm returns a clustering that satisfies the fairness requirements with some additive error. When the maximum number of groups/classes to which a point may belong is $\Delta$, the additive error/violation is at most $4\Delta + 3$; in the most common case of $\Delta=1$, the additive violation is at most $3$. \citeauthor{bercea2019cost} also gave constant factor approximation algorithms for a variety of clustering objective (including $k$-median and $k$-means) that violate the fairness requirement only by a small additive value.
More recently,~\citet{bandyapadhyay2020coresets} designed algorithms that compute a constant factor approximation for fair $k$-median and $k$-means that run in time $(k\Delta)^{O(k\Delta)} \poly(n)$; these algorithms do not violate the fairness constraints. 
However, one of the main questions in the area of fair clustering still remains open:
\begin{quote}
\textit{What is the best polynomial-time approximation algorithm for representation fair clustering?}  
\end{quote}

We also study \textit{exact} fair representation clustering, which is an important special case of the problem.
\begin{definition}[exact fairness]\label{def:exact}
Assume that we are given a set of points $X$ that come from $\ell$ disjoint groups: $X = X_1 \cup \cdots \cup X_\ell$. A $k$-clustering $C_1, \cdots, C_k$ of $X$ is {\em exactly fair} if
\begin{align}\label{eq:exact-fair-cond}
    \forall i\in [k], j\in [\ell],\quad |C_i \cap X_j| = \frac{|X_j|}{|X|} \cdot |C_i| 
\end{align}
We define \textit{fairlet} as a minimal size non-empty set of points that is exactly fair. Note that all fairlets have the same size (when $X$ is fixed). Denote this size by $f$. Further, for every $j\in [\ell]$, let $f_j$ denote the number of points from group $j$ in any fairlet. 
\end{definition}
This notion was previously studied for (1) $k$-center~\citep{rosner18privacy,bercea2019cost} and (2) $k$-clustering with $\ell_p$-objective on balanced instances (instances with $f_j=1$ for all $j$ and $f=\ell$)~\citep{bohm2020fair}. In all these special cases of exact fairness, the fair clustering problem admits a constant factor approximation.  

\paragraph{Other Related Work.}
Fair clustering is an active domain of research and by now it has been studied under various standard notions including both group fairness and individual fairness, e.g., ~\citep{chierichetti2017fair,bera2019fair,huang2019coresets, kleindessner2019fair,jones2020fair,chen2019proportionally, micha2020proportionally,jung2019center,mahabadi2020individual, chakrabarty2021better,vakilian2021improved,brubach2020pairwise, kleindessner2020notion,ghadiri2020fair, abbasi2020fair,makarychev2021approximation,chlamtavc2022approximating,esmaeili2020probabilistic,esmaeili2021fair}.

\subsection{Our Results} In this paper, we study the fair representation $k$-median problem and give an $O(\log k)$-approximation for it that runs in time $n^{O(\ell)}$.
Importantly, we get a polynomial time algorithm for every fixed $\ell$ -- this is the case in most practical settings of interest for fairness applications. In fact, we design an algorithm that can handle arbitrary fairness profiles (see below for the definitions). Further, we design a much faster algorithm for fair $k$-median with exact fairness constraints and small $f$, where $f$ is the size of a fairlet. It runs in time $(kf)^{O(\ell)}\log n + poly(n)$. We emphasize that even in the case $\ell= O(1)$, all previous results on fair representation clustering with more than two protected groups, notably~\cite{bera2019fair,bercea2019cost,bandyapadhyay2020coresets}, either violate the fairness constraints by additive terms or run in time that is exponential in $k$. In this paper, we present first polynomial time approximation algorithms for fair representation with $\ell = O(1)$ multiple protected classes that satisfy the fairness requirement with no additive violations (see Theorem~\ref{thm:main} and Theorem~\ref{thm:exact-fair}).

\paragraph{General Representation Fairness.}
Our main algorithm has three steps: {\em location consolidation}, {\em approximation the metric by a distribution of tree metrics}, and finally {\em solving fair clustering on a tree}. 

In the first step, we run an existing constant factor approximation algorithm for $k$-median to find a set of $k$ centers $\sC=\{c_1,\dots, c_k\}$. Next, we move each point to its closest center $c_i$ in $\sC$ in the constructed solution. In other words, we reduce the initial instance of size $n$ (with possibly $n$ different locations) to an instance of size $n$ with exactly $k$ locations. (i.e., we may have multiple data points mapped/moved to each location.)

In the second step, we use the metric embedding technique by~\citet{fakcharoenphol2004tight} to approximate the reduced instance by a tree metric with expected distortion $O(\log k)$. As the reduced instance, this instance also has at most $k$ different locations. Additionally, the metric on these locations is a tree metric.

Finally, in the third step, we use dynamic programming (DP) to find a fair assignment of $n$ points located at $k$ different locations to $k$ centers. The DP runs in time $n^{O(\ell)}$. 

We note that that our first step is very similar to that used by~\citet{bera2019fair} and~\citet{bercea2019cost}; they first find a not-necessarily-fair clustering of the data points and then reassign the points so as to ensure that the clustering is approximately fair. In the context of $k$-median, the idea of approximating the input metric with a distribution of dominating trees was introduced by~\citet{bartal1998approximating}; this approach was recently used by~\citet{backurs2019scalable} in their approximation algorithm for a different variant of fair representation clustering with 2 groups. Our dynamic programming algorithm is novel and very different from DP algorithms previously used for solving $k$-median on trees (see e.g., \citep{KarivHakimi,Tamir96,Angelidakis17}).

\begin{theorem}\label{thm:main}
There exists a randomized $O(\log k)$-approximation algorithm for fair representation $k$-median that runs in $n^{O(\ell)}$ time.
\end{theorem}
Note that in practice the number of classes is usually {\em a small constant}. Then our algorithm runs in polynomial time. The problem is interesting even when $\ell$ is a fixed single-digit number.

\begin{remark}\label{remark:arbitrary-profile}
Our approach works in a more general setting with a set of {\em fair profiles} $F$, where each profile $p \in F$ is a vector of length $\ell$. A cluster $C$ is fair w.r.t. $F$ if $(|X_1 \cap C|, \dots, |X_\ell \cap C|) \in F$. For example, this general notion captures the setting in which we only need to guarantee that in each cluster, a sufficiently large fraction of members belong to one of the disadvantaged groups specified by $D \subseteq [\ell]$; $\forall C, \sum_{i\in D} |X_i \cap C| \ge \alpha \cdot |C|$. 

To the best of our knowledge, none of the previous results on fair representation clustering implies an approximation bound for this general ``fairness profile'' notion. 

Given a clustering instance and set $F$, our algorithm finds a clustering such that all clusters are fair w.r.t. $F$. Assuming the existence of a membership oracle for $F$ that runs is time $t_F$, the running time and the approximation factor are $t_F \cdot n^{O(\ell)}$ and $O(\log k)$ as in Theorem~\ref{thm:main}. In most natural scenarios the membership oracle can be implemented efficiently and the asymptotic runtime of our algorithm remains $n^{O(\ell)}$. For instance, in the aforementioned fairness requirement which guarantees the presence of disadvantaged groups in each cluster, $t_F = O(\ell)$; hence, the total runtime of our algorithm in this setting is still $n^{O(\ell)}$.
\end{remark}

\begin{remark}
We describe and analyze our algorithm for the case when each point belongs to exactly one group $X_i$. However, with a minor modification,our algorithm can also handle the case when a point may belong to multiple groups or do not belong to any group. We introduce a ``virtual'' group $Y_G$ for every $G\subset [\ell]$ such that there exists a point that belongs to groups $X_i$ with $i \in G$ and only to them. Note that new virtual groups are disjoint and cover $X$. We define the fairness constraints as follows: for every $j\in[\ell]$,
$\alpha_i |C_i| \leq \sum_{G:j\in G} |C_i\cap Y_G| \leq \beta_i|C_i|$.
By Remark~\ref{remark:arbitrary-profile}, our algorithm can handle these constraints; they are equivalent to the original constraints, since $\sum_{G:j\in G} |C_i\cap Y_G| = |C_i \cap X_j|$.
Note that now the algorithm runs in time $n^{O(\ell')}$ where $\ell'$ is the number of virtual groups (clearly $\ell' \leq 2^\ell$ but $\ell'$ may be much smaller than $2^\ell$). 
\end{remark}

\paragraph{Exact Representation Fairness.}
We significantly improve the running time of our algorithms when the fairness constraints are exact (see Definition~\ref{def:exact}) and each point belongs to exactly one group. 
We first run the algorithm by~\citet{bera2019fair} that returns a set of centers and an assignment of points to these centers that ``nearly'' satisfies the fairness requirement. Next, we move each point to its assigned center. We prove that there exists an $O(1)$-approximately optimal fair assignment that only moves a set of $O(kf^2)$ points $S^*$. Lastly, we show that we can find such a set $S$ of size $O(k^2f^2)$ in polynomial time. Then, loosely speaking, we run our main algorithm on the set $S$. Since $S$ has only $O(k^2 f^2)$ points, the algorithm runs in time $(kf)^{O(\ell)}$.  

\begin{theorem}\label{thm:exact-fair}
There exists a randomized $O(\log k)$-approximation algorithm for exactly fair $k$-median that runs in time $(kf)^{O(\ell)} \log n + \poly(n)$, where $f$ is the fairlet size.
\end{theorem}



\section{Preliminaries}\label{sec:prelim}

\paragraph{Embedding into a distribution of dominating trees.} Assume that we are given a metric space $(M, d)$. We will consider trees $\chi$ on $M$ (whose edges have positive lengths) and shortest-path metrics $d_\chi$ that they define on $M$. We say that $(M, d)$ is $\alpha$-approximated by a probabilistic distribution $\sD_{\sX}$ of dominating trees $\chi$ with distortion $\alpha \geq 1$ if 
\begin{itemize}
    \item Every tree metric $d_\chi$ in the support $\sX$ of $\sD_{\sX}$ dominates metric $(M, d)$; i.e., $d_{\chi}(u, v) \geq d(u, v)$ for all $u,v\in M$.
    \item $\E_{\chi\sim \sD_{\sX}}[d_{\chi}(u,v)]\leq \alpha \cdot d(u,v)$ for all $u,v \in M$.
\end{itemize}
A key component in our algorithms is the following result by Fakcharoenphol et al.\  \cite{fakcharoenphol2004tight} (see also \cite{bartal1996probabilistic} by Bartal). \begin{theorem}[\citet{fakcharoenphol2004tight}]\label{thm:prob-metric-embedd}
Every metric space $(M, d)$ can be approximated by a distribution of dominating trees with distortion at most $O(\log |M|)$. Moreover, we can sample from the distribution of dominating trees in polynomial time. 
\end{theorem}

\section{Algorithms}\label{sec:algorithm}
Now we present an $O(\log k)$-approximation algorithm for fair representation $k$-median that runs in time $n^{O(\ell)}$, where $n$ is the number of points and $\ell$ is the number of groups.



In our algorithms, we are going to transform instances to ``simpler'' instances by moving points to new locations. It will be helpful to distinguish between data points and their locations. Our algorithms will not change data points and their memberships in groups, but will change their locations. Formally, we think of abstract data points that are mapped to points/locations in a metric space; our algorithms will modify this mapping between data points and their locations. We call this process a reassignment. We denote the set of data points by $X$ and the set of locations by $L$. Initially $L = X$ and every point $x\in X$ is assigned to location $x$, but this will change when we transform the instance. We denote the location of point $x$ w.r.t. instance $\mathcal{I}$ by $\loc(x) = \loc_{\mathcal{I}}(x)$.

Now, if two data points at the same location belong to the same group, then they are interchangeable for our algorithm. Thus, instead of storing the actual data points, the algorithm will only store the number of points from each group at every location.

Denote ${\sR}_{\geq 0} \equiv \{0,\dots, n\}^{\ell}$ and $\sR \equiv \{-n,\dots, n\}^{\ell}$. For each $q \in L$ and $j \in [\ell]$, let $v_j(q)$ be the number of data points from group $j$ at location $q$. We call vector $v(q) = (v_1(q), \dots, v_\ell(q)) \in {\sR}_{\geq 0}$ the profile of location $q$.
We define the profile of a set of data points $\mathsf{S}$ as a vector $v(\mathsf{S})$ in ${\sR}_{\geq 0}$ whose $j$-th coordinate equals the number of data points from group $j$ in the set; $v(\mathsf{S}) := \sum_{q\in \mathsf{S}} v(q)$.

Consider a set of $k$ clusters $\sC$. To describe a clustering, we introduce an assignment function $r:L\times \sC \to {\sR}_{\geq 0}$. For each center $c\in \sC$, $r_j(q,c)$ denotes the number of points from group $j$ at location $q$ that we assign to $c$. Note that vector $R(c) = \sum_{q\in L} r(q,c) \in {\sR}_{\geq 0}$ specifies the number of points from each group that are assigned to center $c$. We call $R(c)$ the profile of center $c$. 
We require that for every $q \in L, \sum_{c \in \sC} r(q,c) = v(q)$, meaning that each data point at location $q$ belongs to exactly one cluster.

The fairness constraints are defined by a set of fair profiles $F \subset {\sR}_{\ge 0}$. We say that the cluster assigned to center $c$ fairly represents groups if $R(c)\in F$. In fair representation clustering with fairness requirement $\{\alpha_j, \beta_j\}_{j\in [\ell]}$
\begin{quote}
$R(c) \in F$ if and only if $\alpha_j \left\|R(c) \right\|_1 \le R_j(c) \le \beta_j \left\|R(c) \right\|_1$ for all $j\in [\ell]$. 
\end{quote}
We restate the objective of fair representation $k$-median as follows 
\begin{align}
    \sum_{c \in \sC} \sum_{q\in L} d(q,c) \cdot \|r(q,c)\|_1.\label{eq:obj}
\end{align}
Note that $\|r(q,c)\|_1$ is the total number of points at location $q$ assigned to center $c$.

\subsection{From a Clustering to a New Instance}\label{sec:cluster-instance} Consider an instance $\sJ$ and a clustering for $\sJ$ (which is not necessarily fair). As we discussed above, we can define the clustering by specifying (i) a set $\sC$ of $k$-centers and (ii) either a mapping $\phi$ from the set of data points $X$ to $\sC$ or, equivalently, an assignment function $r(q,c)$. The cost of the clustering is $$\mathbf{cost} = \sum_{x\in X} d(\loc_{\sJ}(x), \phi(x)) = \sum_{c \in \sC} \sum_{q\in L} d(q,c) \cdot \|r(q,c)\|_1.$$ 
Let us move every data point from its original location in $\sJ$ to $\phi(x)$. We get a new instance $\sJ'$. Note that $\loc_{{\sJ}'}(x) = \phi(x) \in \sC$. The profile $v'(c)$ of location $c\in \sC$ is $v'(c) = R(c)$ (the profile $v'(x)$ of a location $x\notin \sC$ is a zero vector). Instance ${\sJ}'$ has the same set of fairness profiles $F$ as the original instance $\sJ$. 

\begin{claim}\label{claim:cost:reassignment}
Consider a clustering $({\sC}', \phi')$ of $X$. Denote its cost w.r.t. instances ${\sJ}$ and ${\sJ}'$ by $cost_{\sJ}$ and $cost_{{\sJ}'}$, respectively. Then
$|cost_{{\sJ}} - cost_{{\sJ}'}| \leq \mathbf{cost}$.
\end{claim}
\begin{proof}
Consider a data point $x$ and the center $\phi'(x)$ that it is assigned to by clustering $({\sC}', \phi')$. Then, by the triangle inequality, $|d(\loc_{\sJ}(x),\phi'(x)) - d(\loc_{{\sJ}'}(x),\phi'(x))| \leq d(\loc_{\sJ}(x),\loc_{{\sJ}'}(x)) = d(\loc_{\sJ}(x),\phi(x))$. Adding up this inequality over all data points $x$, we get the desired inequality $|cost_{{\sJ}} - cost_{{\sJ}'}| \leq \mathbf{cost}$.
\end{proof}

\subsection{Location Consolidation Step}\label{sec:location_consolidation}
We run a constant factor approximation algorithm for the standard $k$-median problem on our set of data points; e.g., the algorithm by Charikar et al.~\cite{charikar2002constant} or by Li and Svensson~\cite{li2016approximating}. We get $k$ centers $\bar \sC = \{\bar c_1, \dots, \bar c_k\}$ and a Voronoi assignment $\phi$ of points to the centers.

As described above, we move each data point $x$ to center $\phi(x)$. We denote the original instance by $\sI$ and the obtained instance by ${\sI}'$. We will refer to ${\sI}'$ as \textit{the reduced instance}.

\begin{claim}\label{clm:reduction}
Let $\opt_{\sI}$ and $\opt_{{\sI}'}$ be the costs of optimal fair clusterings for ${\sI}$ and ${\sI}'$, respectively. Assume that we used an $\alpha$-approximation algorithm for $k$-median in the consolidation step. Further assume that there is a $\beta_k$-approximation algorithm for fair representation $k$-median on ${\sI}'$. Then, there is an $(\alpha\beta_k + \alpha + \beta_k)$-approximation algorithm for fair representation $k$-median on $\sI$.
\end{claim}
\begin{proof}
Observe that the cost of an optimal (not necessarily fair) $k$-median clustering of $\sI$ is at most the cost of an optimal \textit{fair} $k$-median clustering of $\sI$. Thus the cost of the clustering that we use in the reduction is at most $\alpha \opt_{\sI}$.

Now consider the optimal fair clustering for $\sI$. By Claim~\ref{claim:cost:reassignment}, the cost of this clustering as a clustering of ${\sI}'$ is at most $(\alpha + 1)\, \opt_{\sI}$. Therefore, $\opt_{{\sI}'} \leq (\alpha + 1)\, \opt_{\sI}$.

Our approximation algorithm for $\sI$ is very straightforward. We simply run the $\beta_k$-approximation algorithm on instance ${\sI}'$ and output the obtained clustering as a clustering of $\sI$. Now we upper bound the cost of this clustering w.r.t. instance ${\sI}'$ and then w.r.t $\sI$. The cost w.r.t. ${\sI}'$ is at most $\beta_k \opt_{{\sI}'} \leq \beta_k(\alpha + 1)\, \opt_{\sI}$. Consequently, by Claim~\ref{claim:cost:reassignment}, the cost w.r.t. $\sI$ is at most 
$(\beta_k(\alpha + 1) + \alpha)\, \opt_{\sI}$, as required.
\end{proof}

\subsection{Embedding into Tree Metrics}
Consider restriction $\bar d$ of metric $d$ to $\bar\sC$. 
We will reduce the problem to the case when $\bar d$ is a tree metric by paying a factor of $O(\log k)$. To this end, we will approximate $\bar d$ by a distribution ${\sD}_\chi$  of dominating trees $\chi$ with distortion $O(\log k)$ and then solve the fair representation $k$-median for a number of trees $\chi$ randomly sampled from ${\sD}_\chi$.

\begin{claim}\label{clm:assignment-tree-metric}
Suppose that there exists an $\alpha$-approximation algorithm $\sA$ for fair representation $k$-median on tree metrics. Then, there is an $O(\alpha \log k)$-approximation algorithm for the reduced problem ${\sI}'$ that succeeds with high probability (i.e., $1 - n^{-c}$ for any desired constant $c$). 
\end{claim}
\begin{proof}
Let $r^*$ be the optimal fair assignment for ${\sI}'$ and $\opt$ be its cost.
Consider an approximation of $\bar d$ by a distribution $\sD_{\sX}$ of dominating trees $\chi$ with distortion $\log k$, which exists by Theorem~\ref{thm:prob-metric-embedd}.

For a tree metric $d_\chi$, let $r_\chi: \bar{\sC} \times \bar{\sC} \rightarrow {\sR}_{\ge 0}$ denote the $O(\alpha)$-approximate assignment of points to centers that algorithm  $\sA$ finds on instance with metric $d_{\chi}$.
Then,
\begingroup
\allowdisplaybreaks
\begin{align*}
    \E_{\chi \sim \sD_{\sX}}[\sum_{q, c\in \bar{\sC}} d(q,c)\cdot \|r_{\chi}(q,c)\|_1] &\leq \E_{\chi \sim \sD_{\sX}}[\sum_{w, c\in \bar{\sC}} d_{\chi}(q,c) \cdot \|r_{\chi}(q,c)\|_1] \\
    &\leq \E_{\chi \sim \sD_{\sX}}[\alpha \cdot \sum_{q, c\in \bar{\sC}} d_{\chi}(q,c) \cdot \|r^*(q,c)\|_1] \\
    &= \alpha \cdot \sum_{q, c\in \bar{\sC}} \E_{\chi \sim \sD_{\sX}}[d_{\chi}(q,c)] \cdot \|r^*(q,c)\|_1 \\
    &\leq \alpha \cdot \sum_{q, c\in \bar{\sC}} O(\log k) \cdot d(q,c) \cdot \|r^*(q,c)\|_1 = O(\alpha \log k)  \cdot \opt.
\end{align*}
\endgroup
The three inequalities above hold, since (i) $d(q,c) \leq d_\chi(q,c)$ (always); (ii) $r_\chi$ is an $\alpha$-approximate assignment/solution;
and (iii) $\mathbb{E}{d_\chi}(q,c) \leq O(\log k) \cdot d(q,c)$ (recall that $\bar{d}$ is a restriction of $d$ to $\bar{C}$; for every $u,v\in \bar{C}, \bar{d}(u,v) = d(u,v)$). 

By Markov's inequality the obtained solution $r_\chi$ approximates the optimal assignment within a factor of $O(\alpha \log k)$ with probability at least $1/2$. By running the proposed algorithm $\Theta(\log n)$ times, we get an $O(\alpha \log k)$-approximate solution w.h.p.
\end{proof}
\subsection{Reduced Assignment Problem on Trees}\label{sec:DP-assignment}

In this section, we assume that $(L, d)$ is a tree metric on $k$ points with profile vector $v(u) \in {\sR}_{\ge 0}$ for every location $u\in L$. We open a center at every location and our goal now is to find a fair assignment of data points to centers. Recall that our notion of fairness is more general than that in Definition~\ref{def:fairness} (we discussed it in Remark~\ref{remark:arbitrary-profile}).

\subsubsection{Conversion to Binary Tree}
We choose an arbitrary root in the tree. It is convenient now to convert the tree to a {\em binary} tree in which every non-leaf vertex has exactly two children. We do that by adding {\em Steiner locations}. Namely, we replace each vertex $u$ with $k > 2$ children with a path of length $k-2$. The first vertex on the path is $u$ (we assume $u_0\equiv u$); other vertices $u_1, \dots, u_{k-2}$ are new Steiner locations. We keep the parent of $u$ unchanged; we let the parent of $u_i$ be $u_{i-1}$ (for $i\geq 1$). For $j \in\{2,\dots, k-1\}$, we reassign $j$-th child $v_j$ of $u$ to $u_{j-1}$ and reassign $k$-th child $v_k$ to $u_{k-2}$. We set the length of all edges on the path $u\to u_1\to \dots\to u_{k-2}$ to 0; we let the length of each edge
$(u_{j-1}, v_j)$ be that of $(u,v_j)$ in the original tree; the length of $(u_{k-2}, v_k)$ to that of $(u, v_k)$ in the original tree. Additionally, we add a Steiner child to vertices that have exactly one child. For a concrete example, consider a tree rooted at vertex $1$, in which $1$ has four children $2,3,4,5$. We transform the tree by adding Steiner nodes 6 and 7 as follows:

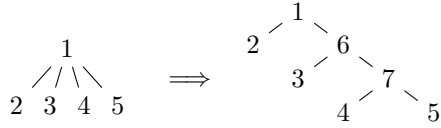
\begin{center}
\begin{tikzpicture}[scale=0.3,baseline=-5mm,level 1/.style = {level distance = 25mm}]
\node {1}
    child {node {2}}
    child {node {3}}
    child {node {4}}
    child {node {5}};
\end{tikzpicture}
\quad $\Huge\displaystyle\Longrightarrow$\quad
\begin{tikzpicture}[scale=0.3,baseline=-10mm,level 1/.style = {level distance = 15mm, sibling distance=40mm}]
\node {1}
    child {node {2}}
    child {node {6}
    child {node{3}}
    child {node {7}
    child {node {4}}
    child {node {5}}
    }};
\end{tikzpicture}
\end{center}

Let $A$ be the set of non-Steiner locations and $B$ be the set of Steiner locations. We open a center at every $a\in A$; we do not open centers at any Steiner location $b\in B$.
For $b\in B$, let  $v(b) = 0$.

\subsubsection{Dynamic Program}
Now we write a dynamic program (DP). We first recall the setup. We are given a binary tree $L$, which contains Steiner and non-Steiner nodes/locations. Each non-Steiner node contains some data points that we want to cluster and each Steiner node contains no data points. Our goal is to move the data points around (assign data points to the non-Steiner nodes) such that the resulting clustering is fair. The objective is to minimize the ``assignment cost'' $\sum_{t,c \in A} d(t,c) \cdot \|r(t,c)\|_1$, where $\|r(t,c)\|_1$ is the number of data points at location $t$ that are assigned to $c$.

Let $T_u$ be the subtree of $L$ rooted at $u$. Let $A_u$ and $B_u$ be non-Steiner and Steiner locations in $T_u$, respectively. For a fixed assignment $\phi$, data points located at nodes/locations in $T_u$ can be classified into two types: \textit{local points} and \textit{out-points}. A \textit{local point} is assigned to a node/location in $T_u$ by $\phi$ and an \textit{out-point} is assigned to a node/location outside $T_u$ by $\phi$. In addition, there are some points outside $T_u$ that are assigned to nodes/locations in $T_u$. We refer to these points as \textit{in-points}. Now we define two functions $\rho_{out}:A_u \to {\sR}_{\geq 0}$ and $\rho_{in}:A_u\to {\sR}_{\geq 0}$, which specify the out-points and in-points, respectively. Namely, for each location $c \in A_u$, $\rho_{out}(c)$ is the profile of out-points located at $c$, and $\rho_{in}(c)$ is the profile of in-points assigned to $c$. Let $q := \sum_{c\in A_u} \rho_{in}(c) - \sum_{y\in A_u} \rho_{out}(y)$ be the ``net-imports'' of $T_u$. Clearly, $q \in {\sR}$.


Now, we create a DP-table $M[u,q]$ where $u\in L$ and $q\in {\sR}$. Loosely speaking, $M[u,q]$ is the cost of the minimum cost partial solution such that (a) clusters for all centers $c\in A_u$ satisfy fairness constraints (note that each cluster consists of the points at locations in $T_u$ assigned to $c$ and in-points assigned to $c$) and (b) the difference between the number of in-points and out-points from group $j$ equals $q_j$. The cost of a partial solution comprises (i) the assignment costs for points in $T_u$ assigned to centers in $T_u$, (ii) for each out-point $x$ located at $y\in A_u$, 
the portion $d(y, u)$ of the assignment cost of $x$ that ``lies'' inside $T_u$, and (iii) for each in-point $x'$ assigned to center $c\in A_u$, the portion $d(u,c)$ of the assignment cost of $x'$ that ``lies'' inside $T_u$. 

Formally, $M[u,q]$ is the cost of the optimal solution for the following problem.

\paragraph{Find}
\begin{itemize}
    \item function $r:A_u \times A_u\rightarrow {\sR}_{\geq 0}$; map $r$ specifies the assignment of data points located in subtree $T_u$ to centers in $T_u$. Namely, $r(y,c)$ is the profile of the set of data points located at $y$ that are  assigned to center $c$. Let $R_u(c) = \sum_{y\in T_u} r(y,c)$. Note that $R_u(c)$ is the profile of the set of data points at locations in $A_u$ that are assigned to center $c$. 
    
    \item function $\rho_{out}:A_u \to {\sR}_{\geq 0}$; map $\rho_{out}$ specifies the set of data points in subtree $T_u$ that are not assigned to any center in $T_u$; these data points will have to be assigned to centers outside of $T_u$ later. Namely, $\rho_{out}(y)$ is the profile of currently unassigned data points at location $y$.

    \item function $\rho_{in}:A_u\to {\sR}_{\geq 0}$; $\rho_{in}$ specifies how many data points located outside of the subtree $T_u$ are assigned to centers in $T_u$. Namely, $\rho_{in}(c)$ is the profile of data points that lie outside of $T_u$ but assigned to center $c$.
\end{itemize}

\paragraph{Such that}
\begin{itemize}
    \item for every $y\in A_u$, $\sum_{c\in A_u} r(y, c) + \rho_{out}(y) = v(y)$. This condition says that every data point at location $y$ is assigned to some center $c \in A_u$ or is currently unassigned.
    
    \item $R_u(c) + \rho_{in}(c) \in F$ for all $c\in A_u$. This condition says that the profile of the cluster centered at $c$ is in $F$ (that is, the cluster satisfies the fairness constraints). Here, $R_u(c)$ counts data points in $T_u$ that are assigned to $c$ and $\rho_{in}(c)$ counts data points outside of $T_u$ that are assigned to $c$.
    
    \item $\sum_{c\in A_u} \rho_{in}(c) - \sum_{y\in A_u} \rho_{out}(y) = q$.
\end{itemize}

\paragraph{Cost} The objective is to minimize the following cost
\begin{gather*}
    \sum_{y,c\in A_u} d(y,c) \cdot \|r(y,c)\|_1 
    + \sum_{c\in A_u} d(c,u) \cdot \|\rho_{in}(c)\|_1
    + \sum_{y\in A_u} d(y,u) \cdot \|\rho_{out}(y)\|_1.
\end{gather*}
If there is no feasible solution the cost is $+\infty$.



\subsubsection{Solving the DP} 
We fill out the DP table starting from the leaves and going up to the root (bottom-up). 

Computing $M[u,q]$ for leaves $u$ is straightforward using the definition of $M[u,q]$. In this case, the subtree $T_u$ rooted at $u$ is a single node. Thus, for any given $q$, the profile at $u$ after assignment is
\begin{align*}
    R_u(u) + \rho_{in}(u) 
    = r(u,u) + \rho_{in}(u)
    = r(u,u) + \rho_{out}(u) + q
    = v(u) + q.
\end{align*}
If $u$ is not a Steiner node, then $M[u,q] = 0$ if $v(u) + q \in F$ and $M[u,q] = \infty$ if $v(u) + q \not \in F$. If $u$ is a Steiner node, then $M[u,q] = 0$ if $q  = 0$ and $M[u,q] = \infty$ if $q \not = 0$.

Now assume that $u$ is not a leaf; since the tree is a binary tree, $u$ has two children $y$ and $z$. Note that node $u$ will send $q_y$ points to $T_y$ and $q_z$ points to $T_z$ for some $q_{y}, q_{z}\in {\sR}$. Then $M[u,q]$ is the minimum over $q_{y}, q_{z}\in {\sR}$ satisfying
\begin{itemize}
    \item if $u\in A$, then $v(u) + q - q_y - q_z\in F$; that is, cluster centered at $u$ satisfies fairness constraints.
    \item if $u\in B$, then $q - q_y - q_z = 0$, since $v(u) = 0$ and no point will be assigned to $u$.
\end{itemize}
of the following cost
\begin{align}\label{eq:dp-update-rule}
    M[y, q_y] + M[z, q_z] + d(u, y) \cdot \|q_y\|_1 + d(u, z) \cdot \|q_z\|_1
\end{align}
The pseudo-code of the DP algorithm is provided in Algorithm~\ref{alg:DP}.

\begin{lemma}\label{lem:DP}
The described dynamic programming algorithm runs in time $n^{O(\ell)}$ and finds an optimal fair assignment of the data points $X$ located in the set of locations $L$. 
\end{lemma}
\begin{proof}
It is straightforward that the DP algorithm correctly computes the DP entries. The cost of the optimal fair assignment for our problem is given by $M[root, \boldsymbol{0}]$ where $root$ is root of the tree and $\boldsymbol{0}$ is the all-zero vector. 

The update rule, which is specified by Eq.~\eqref{eq:dp-update-rule} and the constraints on $q_y, q_z, q$ and $v(u)$, computes $M[u,q]$ correctly in time $n^{O(\ell)}$, which accounts for enumerating over all possible values of $q_y$ and $ q_z$. Since the DP tables contains $k \cdot n^{O(\ell)}$ cells, the total time to fully compute the table is $k \cdot n^{O(\ell)} = n^{O(\ell)}$.  
Finally, once the whole table is computed, we can recover the assignment itself by traversing the table from $M[root, \boldsymbol{0}]$ in the reverse direction, which takes $n^{O(\ell)}$ time.
The total running time is $n^{O(\ell)}$.
\end{proof}
\begin{proof}[\sc Proof of Theorem~\ref{thm:main}:]
We first apply the consolidation step and thus reduce the problem to the set of locations $\bar{\sC}$ is of size $k$. Then we approximate metric on $\bar{\sC}$ by a distribution of dominating tree metrics, as explained in \ref{clm:assignment-tree-metric}.
We obtain a logarithmic number of instances with tree metrics.
We exactly solve each of them using the DP algorithm described above.
Finally, we return the best of the clusterings we find.
It follows from Claims~\ref{clm:reduction} and \ref{clm:assignment-tree-metric} and Lemma~\ref{lem:DP}, that the algorithm finds an $O(\log k)$ approximation with high probability in time $n^{O(\ell)}$.
\end{proof}

\begin{algorithm}
\caption{DP Algorithm on Trees} \label{alg:DP}
\begin{algorithmic}[1]
\REQUIRE $T$ is a binary tree of height $d$, $r$ is the root of $T$, and for each $i = 0,1,\dots,d$, $N_i$ is the set of nodes in the $i$-th level of $T$.
\FOR{$u \in N_d$}
    \FOR{$q \in {\sR}$}
        \IF{$u$ is a non-Steiner node}
            \IF{$v(u) + q \in F$}
                \STATE $M[u,q] = 0$
            \ELSE
             \STATE  $M[u,q] = \infty$
            \ENDIF
        \ELSE
            \IF{$q = 0$}
                \STATE $M[u,q] = 0$
            \ELSE
                \STATE $M[u,q] = \infty$
            \ENDIF
        \ENDIF
    \ENDFOR
\ENDFOR
\FOR{$i = d-1,d-2,\dots,0$}
    \FOR{$u \in N_i$}
        \STATE $y,z \gets \texttt{children of} \hspace{2mm} u$
        \IF{$u$ is a Steiner node}
            \STATE $M[u,q] = \min \{M[y, q_y] + M[z, q_z] + d(u, y) \cdot \|q_y\|_1 + d(u, z) \cdot \|q_z\|_1 : q_y,q_z \in {\sR}, q = q_y+q_z \}$
        \ELSE
            \STATE $M[u,q] = \min \{M[y, q_y] + M[z, q_z] + d(u, y) \cdot \|q_y\|_1 + d(u, z) \cdot \|q_z\|_1 : q_y,q_z \in {\sR}, v(u) + q - q_y - q_z \in F \}$
        \ENDIF
    \ENDFOR
\ENDFOR
\RETURN{} $M[r,0]$
\end{algorithmic}
\end{algorithm}
\section{Special Case of Exact Fairness}\label{sec:exact}
In this section, we design an $O(\log k)$ approximation algorithm for $k$-median with exact fairness constraints. The algorithm runs in time $(kf)^{O(\ell)} \log n + \poly(n)$ and succeeds with high probability.

To recall, we say that a cluster centered at $c$ is {\em exactly fair} if $\forall j \in [\ell]$, $R_j(c) = \frac{|X_j|}{|X|} \| R(c) \|_1$, 
where $\frac{|X_j|}{|X|}$ is the proportion of data points in $X$ that belongs to group $j$. We say that a clustering is \textit{exactly fair} if all of its clusters are exactly fair. More generally, a set of data points $S$ with profile $v(S) \in \mathcal{R}_{\geq 0}$ is \textit{exactly fair} if for all $j \in [\ell]$,
$v_j(S) = \frac{|X_j|}{|X|} \| v(S) \|_1$.

Recall that we defined \textit{fairlet} in Definition~\ref{def:exact} as a minimal size non-empty set that is exactly fair. Note that all fairlets have the same size and we use $f$ to denote their size. Further, for every $j\in [\ell]$, we use $f_j$ to denote the number of points from group $j$ in any fairlet. 


\subsection{Reassignment Method} 
We start with the reduced instance $\sI'$ on the set of locations $\bar \sC$, which we constructed in the location consolidation step.

Given any subset $\mathsf{S} \subseteq X$ with profile vector $v(\mathsf{S})$, we say that $v(\mathsf{S})$ is $\gamma$-approximately fair if
\begin{align}\label{eq:gamma-approx-fair}
 \Bigl| v_j(\mathsf{S}) - \frac{f_j}{f}\| v(\mathsf{S})\|_1 \Bigr| \leq \gamma, \;\;\forall j\in [\ell].
\end{align}
Given an assignment $r$, we say that $r$ is a $\gamma$-approximately fair assignment if $R(c)$ is $\gamma$-approximately fair for all $c \in \sC$, where $R(c) = \sum_{x\in X} r(x,c) \in {\sR}_{\geq 0}$ is the center profiles corresponding to $r$.


We use the following result by Bera et al.\ \cite{bera2019fair}.

\begin{theorem}\label{thm:additive-assignment}
For any metric space $X$ and a set of centers $\sC \subseteq X$, There exists a $3$-approximately fair assignment $r: X \times \sC \rightarrow {\sR}_{\ge 0}$ whose cost is no more than the optimal fair assignment of $X$. Moreover, there is an algorithm which finds this assignment in polynomial time.
\end{theorem}

We start by running the assignment algorithm from Theorem~\ref{thm:additive-assignment} on instance $\sI'$. 
Let $r$ denote this new assignment and $R(c) = \sum_{x\in X} r(x,c) \in {\sR}_{\geq 0}$ be the center profiles corresponding to $r$.
As discussed in Section~\ref{sec:cluster-instance}, every assignment defines a new instance. Let ${\sI}''$ be the instance defined by assignment $r$.
Now our goal is to modify $r$ by reassigning only a small number of points so that the obtained assignment is fair. We describe this reassigning procedure below.



\begin{lemma}\label{lem:reassignment}
Assume that there is a $\beta_k$-approximation algorithm for the fair clustering problem on $\sI''$. Then, there is an $(2\beta_k + 1)$-approximation algorithm for $\sI'$.
\end{lemma}

\begin{proof}
The proof is similar to that of Claim~\ref{clm:reduction}. We run the $\beta_k$ approximation algorithm on instance ${\sI}''$ and output the obtained clustering as a clustering for ${\sI}'$. We now upper bound its cost.

Let $\opt_{\sI'}$ and $\opt_{\sI''}$ be the costs of the optimal fair clusterings/assignments for instances $\sI'$ and $\sI''$, respectively. 
By Claim~\ref{claim:cost:reassignment} applied to $r$,
$\opt_{\sI''} \leq 2\opt_{\sI'}$. Thus the cost of the clustering we find w.r.t. ${\sI}''$ is at most $2\beta_k\opt_{\sI'}$. Now applying Claim~\ref{claim:cost:reassignment} to this clustering, we get that its cost w.r.t. ${\sI}'$ is at most $(2\beta_k + 1)\opt_{\sI'}$, as required.
\end{proof}

\paragraph{Algorithm for the reassignment problem.}
We present an $O(\log k)$-approximation algorithm for the reassignment problem that runs in time $(fk)^{O(\ell)}$. 

Suppose we are given an instance $\sJ$ with locations $c_1, \dots, c_k$, metric $d$ and profile vectors $v(c_1), \dots, v(c_k)$, that satisfy $3$-approximately fair requirement as in Eq.~\eqref{eq:gamma-approx-fair}. 
We will apply our reassignment algorithm to $\sJ = \sI''$.
For every $i \in [k]$, let $C_i$ be the set of points assigned to $c_i$ in $\sJ$. 
\begin{definition}
For every $i \in [k]$, let $F_i \subseteq C_i$ be a maximal exactly fair subset of $C_i$. Decompose each $F_i$ into $n_i = |F_i|/f$ fairlets ${\sF}^{(i)}_1, \dots, {\sF}^{(i)}_{n_i}$. Let $P_i = C_i \setminus F_i$ be the set of remaining points. We say that ${\sD} = \{({\sF}^{(i)}_1, \dots, {\sF}^{(i)}_{n_i}, P_i)\}_{i\in [\ell]}$ is a fairlet decomposition for clustering $C:= (C_1, \dots, C_k)$. We call points in $P_i$ problematic points. 
\end{definition}
\begin{lemma} \label{lem:problematic}
Consider a fairlet decomposition for $C$. Then for every $i \in [k]$, $|P_i| < 4f$.
\end{lemma}

\begin{proof}
Consider an arbitrary $i \in [k]$. We first assume that $n_i = 0$ ; in other words, $F_i$ is empty. Let $v(C_i)$ be the profile vector of $C_i$. Since $F_i = \emptyset$, $P_i = C_i$ and there exits a group $j \in [\ell]$ such that $v_j(C_i) < f_j$. 
Further, Since $P_i = C_i$ is $3$-approximately fair, $|v_j(C_i) - \frac{f_j\cdot |P_i|}{f}| \leq 3$,
which implies that $\frac{f_j}{f} |P_i| \leq v_j(C_i) + 3 < f_j + 3$. Dividing both sides of the above equation by $f_j$,     $\frac{|P_i|}{f} < \frac{3}{f_j} + 1 \leq 3+1 = 4$,
since $f_j \geq 1$. Thus, $|P_i|<4f$.

Now, we drop the condition $n_i = 0$. Let $v(P_i)$ be the profile vector for $P_i$. We show that not only $v(C_i)$ but also $v(P_i)$ satisfies Eq.~\eqref{eq:gamma-approx-fair} with $\gamma = 3$. Indeed, for each $j \in [\ell]$, $v_j(C_i) - v_j(P_i) = n_i f \frac{f_j}{f} = (\| v(C_i) \|_1 - \| v(P_i) \|_1) \frac{f_j}{f}$.
Thus, for each $j \in [\ell]$, $    \Bigl|v_j(P_i) - \frac{f_j}{f} \| v(P_i) \|_1\Bigr| = \Bigl|v_j(C_i) - \frac{f_j}{f} \| v(C_i) \|_1\Bigr| \leq 3$.
Hence the same argument as for the case $n_i = 0$ shows that $|P_i| < 4f$.
\end{proof}
Next, we give an outline of our approach. 
\begin{enumerate}
    \item First, we show that there exists an exactly fair assignment $r$ for $\sJ$ of cost at most $2 \opt_{\sJ}$ that only reassigns $O(k f^2)$ points.\label{item:first}
    \item Next, we show that it is possible to find a subset $S_i \subseteq C_i$ for each $i \in [k]$ such that $|S_i| = O(k f^2)$ for all $i \in [k]$, and there exists a solution of cost at most $2 \opt_{\sJ}$ that only reassigns points in $S := \bigcup_{i \in [k]} S_i$. \label{item:second}
    \item Finally, we apply our dynamic programming approach to reassign points in $S$. Since $|S|=O(k^2 f^2)$, the DP algorithm runs in  $(kf)^{O(\ell)}$ time.\label{item:third}
\end{enumerate}

To show~(\ref{item:first}), we impose an additional assumption on the structure of the desired (re)assignment $r$. Let $C_1', \dots, C_k'$ be the clustering resulted from an exactly fair assignment of $\sJ$, where $C_i'$ is the cluster of points at location $c_i$ after the reassignment.

Let $C_1, \dots, C_k$ be the clustering before the (re)assignment. Now, we construct fairlet decompositions for $(C_1, \dots, C_k)$ and $(C_1', \dots, C_k')$. Let $E \subseteq X$ be the set of data points that are assigned to the same center before and after (re)assignment $r$. For each $i \in [k]$, let $E_i = \{x \in E: \loc_{\sJ}(x) = c_i\}$ be the set of data points in $C_i$ that are fixed by $r$. 

For every $i \in [k]$, let $F'_i \subseteq E_i$ be a maximal exactly fair subset of $E_i$. We decompose $F'_i$ into $n'_i = |F'_i|/f$ fairlets. We call this decomposition ${\sD}_E$.
Now, we greedily extend ${\sD}_E$ to a fairlet decomposition for $(C_1, \dots, C_k)$. Let $P_i$ be the set of problematic points in the obtained decomposition. Note that Lemma~\ref{lem:problematic} applies to $P_i$ and therefore $|P_i| < 4f$.

Similarly, we greedily extend decomposition ${\sD}_E$ to a fairlet decomposition for $(C_1', \dots, C_k')$. We refer to fairlets that we add to ${\sD}_E$ as \textit{new fairlets}. Let $\sN$ be the set of all new fairlets.

\begin{definition}[restricted assignment]
In the restricted assignment problem, we need to find a minimum cost assignment $r: \{c_1,\dots,c_k\} \times [k] \rightarrow {\sR}_{\geq 0}$ of $\sJ$ such that 
\begin{itemize}
    \item $r$ is a valid exactly fair assignment. Formally,
    \begin{align*}
        &\forall c \in\{c_1, \cdots, c_k\}, &&\sum_{i=1}^k r(c, c_i) = v(c) \\
        &\forall j \in [\ell], c \in\{c_1, \cdots, c_k\}, &&R_j(c) = \frac{f_j}{f} \| R(c) \|_1
    \end{align*}
    \item every new fairlet $\sF$ in $C_i'$ contains at least one point from $C_i$.
\end{itemize}
\end{definition}

\begin{lemma} \label{lem:restricted}
Let $\opt^R_{\sJ}$ and $\opt_{\sJ}$ denote the optimal values of the restricted and vanilla fair assignment problems with input $\sJ$ respectively. Then, $\opt^R_{\sJ} \leq 2 \opt_{\sJ}$.
\end{lemma}
\begin{proof}
Let $r$ be an optimal fair assignment of $\sJ$. Let $C_1', \dots, C_k'$ be the resulting clusters; the set $C_i'$ consists of the points assigned to $c_i$ by $r$. Now, suppose that there exits a new fairlet $\sF$ in some $C_i'$ ($i \in [k]$) such that $\sF \cap C_i = \emptyset$. Without loss of generality, assume that the points in $\sF$ are from $C_{j_1}, C_{j_2}, \dots, C_{j_m}$ for some $m \in [k]$. Let $w_t := |C_{j_t} \cap \sF|$ for every $t\in [m]$ and $W = \sum_{t=1}^m w_t$. Then, the total cost of forming $\sF$ is $\sum_{t = 1}^m w_t \cdot d(c_{j_t}, c_i)$.
Now, for every $s\in [m]$, compute the cost of moving all points in fairlet $\sF$ to cluster $C_{j_s}$ (i.e., reassign all data points in $\sF$ to center $c_{j_s}$) and denote it as $M_s$.
\begin{equation*} 
\begin{split}
    M_S = \sum_{t = 1}^m w_t d(c_{j_t}, c_{j_s}) & \leq \sum_{t \in [m]\setminus\{s\}} w_t (d(c_{j_t}, c_i) + d(c_{j_s},c_i)) 
     = \Bigg(\sum_{t \in [m]\setminus\{s\}} w_t \Bigg) d(c_{j_s},c_i) 
     + \sum_{t \in [m]\setminus\{s\}} w_t d(c_{j_t}, c_i).
\end{split}
\end{equation*}
Next, consider the convex combination $\sum_{s=1}^m \frac{w_s}{W} M_s$. 
\begin{align*}
\sum_{s=1}^m \frac{w_s}{W} M_s =\sum_{s = 1}^m \frac{w_s}{W} \sum_{t = 1}^m w_t d(c_{j_t}, c_{j_s}) \leq  \sum_{s = 1}^m \frac{w_s}{W} \Bigg(\sum_{t \in [m]\setminus \{s\}} w_t \Bigg) d(c_{j_s},c_i) 
    + \sum_{s = 1}^m \frac{w_s}{W} \sum_{t \in [m]\setminus\{s\}} w_t d(c_{j_t}, c_i) 
    < 2 \sum_{t = 1}^m w_t d(c_{j_t}, c_i).
\end{align*}
This implies that there exist $s^* \in [m]$ with $C_{j_{s^*}} \cap \sF \neq \emptyset$ such that the cost $M_{s^*}$ of reassigning $\sF$ to $c_{j_{s^*}}$ is at most twice the current assignment cost of $\sF$.

We perform this fairlet reassignment procedure to all new fairlets that do not satisfy the restricted assignment property and obtain a new restricted assignment $r'$. 
Our argument shows that the assignment cost of $r'$ is at most twice the assignment cost of $r$; hence,
\begin{align*}
    \opt^R_{\sJ} \le 2\opt_{\sJ}.
\end{align*}
\end{proof}

Next, we show that the optimal restricted assignment is even more structured.

\begin{lemma} \label{lem:structure1}
There exists an optimal solution $r'$ to the restricted assignment of $\sJ$ such that 
\begin{enumerate}[label=(\alph*)]
    \item for any $i\in [k]$ and any new fairlet $\sF \subseteq C_i'$, there exists a problematic point $x \in \sF \cap C_i$;
    \item if $r'$ moves a point $x \in \sF'$ for some fairlet $\sF' \subseteq C_i$ and $i \in [k]$, then $r'$ moves all points in $\sF'$.
\end{enumerate}
\end{lemma}
\begin{proof}
Let $r$ be the optimal solution for the restricted assignment problem. We apply two transformations to it.

\paragraph{Transformation I: Simplifying the reassignment multigraph.} For every group $j\in[\ell]$, define the reassignment multigraph $G_j$ on the set of centers $c_1,\dots, c_k$. There is a \textit{directed} edge from $c_a$ to $c_b\neq c_a$ in $G_j$ for every point $x$ that is reassigned from $c_a$ to $c_b$ by $r$. We label such an edge with label $x$. Note that there are $r_j(c_a, c_b)$ parallel edges from $c_a$ to $c_b$ in $G_j$. 

Assume that there is a vertex $c$ that has both incoming and outgoing edges in $G_j$. Let $(c', c)$ be an incoming edge and $(c,c'')$ be an outgoing edge (it is possible that $c'=c''$). Denote their labels by $x$ and $x'$. By the definition of the graph, $r$ reassigns point $x\in X_j$ from $c'$ to $c$ and point $x'\in X_j$ from $c$ to $c''$. We modify $r$ as follows (see Figure~\ref{fig}).
\begin{figure}[b]
\centering
\begin{tabular}{l|ccc}
&\textbf{original assignment $r$}&& \textbf{new assignment $\tilde r$}\\ \hline
assignment& $\xymatrix @C=3pc {
& c \ar[dr]^{x'}&
\\
c'\ar[ur]^{x} & & c''}$  & \qquad\qquad &$\xymatrix @C=3pc {
& c \ar@{..>}@(dl,dr)^{x'}&
\\
c'\ar[rr]^{x} & & c''}$ \\ \\
cost& $d(c',c) + d(c,c'')$ & & $d(c', c'')$
\end{tabular}
\caption{We reassign $x$ and $x'$. When we do the reassignment, we replace two edges $(c',c)$ and $(c,c'')$ with one edge $(c', c'')$ in graph $G_j$. Note that there are no loops in $G_j$. We have drawn a dotted arrow from $c$ to itself in the figure simply to indicate that $\tilde r$ does not reassign $x'$.}\label{fig}
\end{figure}
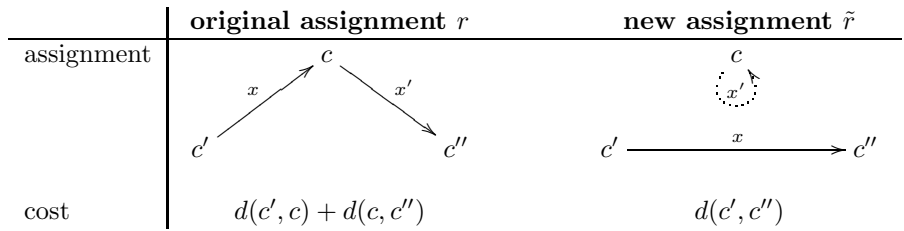
We only change the assignment of points $x$ and $x'$: we reassign $x$ to $c''$ and $x'$ to~$c$. Denote the obtained assignment $\tilde r$. Since $x$ and $x''$ belong to the same group $X_j$, we do not change the profiles of any clusters and thus $\tilde r$ is exactly fair. Further, if $r$ assigns a point $y$ from $c_i$ to $c_i$ then so does $\tilde r$. Therefore, $\tilde r$ is also a solution to the \textit{restricted} assignment problem. Finally, its cost is at most that of $r$: indeed the assignment costs of all points other than $x$ and $x'$ do not change; the total assignment cost of $x$ and $x'$ was equal to $d(c',c) + d(c,c'')$ for $r$ and is equal to $d(c',c'')$ for $\tilde r$; we have 
$d(c',c) + d(c,c'') \geq d(c',c'')$ by the triangle inequality.

We perform this step over and over until there are no vertices that simultaneously have incoming and outgoing edges. (Note that each time we perform this step, the number of edges in one of the graphs $G_j$ decreases by 1 and does not change in all other graphs $G_{j'}$ (for every $j'\in [\ell]\setminus\{j\}$). So we will necessarily stop at some point.)

Denote the obtained assignment by $r'$. It is a fair assignment and its cost is at most that of $r$. Consider the reassignment multigraphs $\{G_j'\}_{j\in[\ell]}$ for $r'$. Each vertex $c$ in $G_j'$ is either a source, sink, or isolated vertex. 

\paragraph{Transformation II: Specifying which points $r'$ moves.} We go over all $i\in [k]$ and $j\in [\ell]$ such that $c_i$ is a source-vertex in $G_j'$. Assignment $r'$ moves $M = \sum_{i'\neq i} r_j(c_i,c_{i'})$ points in group $j$ from $C_i$ to other clusters. Since all points in $C_i \cap X_j$ are interchangeable, we are free to choose any subset of $M$ points in $C_i \cap X_j$ to move. We choose this set in a special way. Let $\sF_1, \dots, \sF_{n_i}$ be fairlets in the fairlet decomposition we constructed for $C_i$. We order all points in $C_i \cap X_j$ as follows: first points from $P_i \cap X_j$, then from $\sF_1 \cap X_j$, then from $\sF_2\cap X_j$, then from $\sF_3\cap X_j$, etc (we order points inside each of these sets arbitrarily). We choose the subset of the first $M$ points w.r.t. this order and assume that $r'$ moves them. Note that $r'$ moves non-problematic points from $C_i \cap X_j$ to another cluster only if it moves all problematic points from $C_i \cap X_j$.

We have defined $r'$. Now we prove that it satisfies properties (a) and (b). To this end, we consider a cluster $C_i$ and analyze the following two cases.
\paragraph{Case 1: Assignment $r'$ does not move any non-problematic points from cluster $C_i$ to other clusters. All points (if any) it moves from $C_i$ to other clusters are problematic.} Then $C_i\setminus P_i \subseteq C_i'$ and thus all fairlets present in the decomposition for $C_i$ are also present in that for $C_i'$. Hence, every non-problematic point belongs to the same fairlet in $C_i$ and $C_i'$. In particular, only problematic points may belong to new fairlets. On the other hand, since $r'$ is a solution to the \textit{restricted} assignment problem, at least one point $x$ in each new fairlet $\sF$ in $C_i'$ must be from $C_i$. It follows that this point must be problematic. We proved that property (a) holds in this case. Since $r$ does not move any points from fairlets ${\sF}' \subseteq C_i$ to other clusters, property (b) trivially holds.

\paragraph{Case 2: Assignment $r'$ moves at least one non-problematic point from $C_i$ to another cluster.} 
Let $t$ be maximal index such that $r'$ moves points from fairlet $\sF_t$ to other clusters. Since we are guaranteed that $r'$ moves at least one non-problematic point, $t$ is well defined.

By our choice of $t$, $r'$ does not move any points from fairlets $\sF_{t+1}, \dots, \sF_{n_i}$ and therefore these fairlets are still present in $C_i'$. Now we show that there are no other fairlets in $C_i'$; in particular, there are no new fairlets.

Consider a point $x$ in $\sF_t$ that $r'$ moves to some other cluster $C_{i'}$. Assume that $x \in X_j$. Then edge $(c_i, c_{i'})$ -- an edge outgoing from $c_i$ -- is present in $G_j'$. Hence, $c_i$ cannot be a sink or isolated vertex and must be a source-vertex in $G_j'$. Therefore, $r'$ does not move any point $y\in X_j$ from another cluster to $C_i$. Also, all points in $C_i \cap X_j$ that precede $x$ w.r.t. the order we considered in the Transformation II step are moved to other clusters by $r'$; in particular, all points in $P_i \cap X_j, \sF_1 \cap X_j,  \dots, \sF_{t-1} \cap X_j$ as well as $x$ itself are moved to other clusters by $r'$.
Therefore, $C_i'\cap X_j$ consists of points from $\sF_{t+1}\cap X_j, \dots, \sF_{n_i} \cap X_j$ and hypothetically some points from $\sF_t \cap X_j$. However, point $x$ from $\sF_t \cap X_j$ is assigned to another cluster and thus there are not enough points from group $j$ left in $\sF_t$ to form another fairlet. We conclude that $\sF_{t+1}, \dots, \sF_{n_i}$ are the only fairlets in $C_i'$. Since $C_i'$ satisfies the exact fairness constraints, each point in $C_i'$ lies in a fairlet, that is, in one of these fairlets. Now (a) trivially holds since there are no new fairlets in $C'$; (b) holds since $r'$ moves all points from fairlets $\sF_1, \dots, \sF_t$ and no points from fairlets $\sF_{t+1},\dots, \sF_{n_i}$.
\end{proof}

Lemma~\ref{lem:structure1} immediately implies the following result which proves~\ref{item:first} and \ref{item:second}.

\begin{corollary} \label{cor:structure9}
Suppose that $\sJ$ is a $3$-approximately fair instance. 
There exists an assignment $r'$ of cost at most $2 \opt_{\sJ}$ that only moves $O(k f^2)$ points. Furthermore, we can choose a subset $S$ of data points such that $|S| = O(k^2 f^2)$ and the solution $r'$ only moves points in $S$.
\end{corollary}

\begin{proof}
Let $r'$ be an optimal restricted assignment that satisfies both conditions in Lemma~\ref{lem:structure1}. In the clustering defined by $r'$, every new fairlet contains a problematic point from the original cluster of the center they are assigned to by $r'$.
Since by Lemma~\ref{lem:problematic} there are only $4kf$ problematic points in total, there are at most $4kf$ many new fairlets. Thus, $r'$ moves at most $4kf^2$ many points. Moreover, by Lemma~\ref{lem:restricted}, the cost of $r'$ is at most $2 \opt_{\sJ}$.

Since we know that $r'$ moves at most $4kf^2$ many points and $r'$ satisfies condition $(b)$ in Lemma~\ref{lem:structure1}, $r'$ would move at most $4kf$ fairlets in each cluster. Recall that each $C_i$ contains $n_i$ fairlets. From each cluster $C_i$, we pick $\min(n_i, 4kf)$ many fairlets and add them to the set $S$. Then, we add all the problematic points to $S$. Then, $|S| \leq 4k^2f^2 + 4kf$ and we can assume that $r'$ only moves the points in $S$.
\end{proof}

Now, we can solve the assignment problem by passing the set $S$ as input to our dynamic programming approach described in Section~\ref{sec:DP-assignment}.

\begin{proof}[{\sc Proof of Theorem~\ref{thm:exact-fair}.}]
Let $\sI$ be an instance of the exactly fair $k$-median problem. First, we perform location consolidation on $\sI$ to obtain the instance $\sI'$ as described in Section~\ref{sec:location_consolidation}. Then, we run the assignment algorithm by~\cite{bera2019fair} on $\sI'$. Let $\sI''$ be the resulting instance. By Theorem~\ref{thm:additive-assignment}, $\sI''$ is $3$-approximately fair. By Corollary~\ref{cor:structure9}, we can pick a subset $S$ of data points such that $|S| = O(k^2f^2)$ and there exists an solution $r'$ of cost at most $2 \opt_{\sI''}$ that only moves points in $S$. Now, we apply the dynamic programming algorithm in Theorem~\ref{thm:main} on $S$ to obtain a solution of cost at most $O(\log k) \opt_{\sI''}$ with high probability that runs in $|S|^{O(\ell)} \log n = (kf)^{O(\ell)} \log n$ time. Since all the previous steps take $\poly(n)$ time, the total running time of this algorithm is $\poly(n) + (kf)^{O(\ell)} \log n$. 

Now, by Lemma~\ref{lem:reassignment} and Claim~\ref{clm:reduction}, the total cost of our solution is $O(\log k) \opt_{\sI}$. 
\end{proof}

\section{Conclusions}
In this paper, we study the fair $k$-median problem. We present an $O(\log k)$-approximation algorithm that runs in time $n^{O(\ell)}$. We further showed that our algorithm works in a more general setting where the fairness requirements are specified as an arbitrary set of fair profiles. This notion ``profile-based fairness'' captures a richer class of fairness requirements that cannot be handled by the previously known approaches for fair representation clustering. 

In addition, in the special case of exact fairness, we present an $O(\log k)$-approximation algorithm that runs in $(kf)^{O(\ell)}\log n + \poly(n)$ time, where $f$ is the size of a fairlet. 

Our paper shows that there exists approximation algorithms for fair representation clustering with $O(1)$ protected classes that run in polynomial time. It remains as an exciting question whether polynomial time $O(\log k)$-approximation can be achieved for the problem when $\ell = \Omega(1)$.     
An approach that can be potentially helpful to resolve the previous question is to extend the ``reassignment'' method used in Section~\ref{sec:exact} to the more general setting of representation fairness (as in Section~\ref{sec:algorithm}).


\bibliographystyle{abbrvnat}
\bibliography{fair-cls}

\begin{thebibliography}{41}
\providecommand{\natexlab}[1]{#1}
\providecommand{\url}[1]{\texttt{#1}}
\expandafter\ifx\csname urlstyle\endcsname\relax
  \providecommand{\doi}[1]{doi: #1}\else
  \providecommand{\doi}{doi: \begingroup \urlstyle{rm}\Url}\fi

\bibitem[Abbasi et~al.(2021)Abbasi, Bhaskara, and
  Venkatasubramanian]{abbasi2020fair}
M.~Abbasi, A.~Bhaskara, and S.~Venkatasubramanian.
\newblock Fair clustering via equitable group representations.
\newblock In \emph{Proceedings of the 2021 ACM Conference on Fairness,
  Accountability, and Transparency}, page 504–514, 2021.

\bibitem[Ahmadian et~al.(2019)Ahmadian, Epasto, Kumar, and
  Mahdian]{ahmadian2019clustering}
S.~Ahmadian, A.~Epasto, R.~Kumar, and M.~Mahdian.
\newblock Clustering without over-representation.
\newblock In \emph{Proceedings of the SIGKDD International Conference on
  Knowledge Discovery \& Data Mining}, pages 267--275, 2019.

\bibitem[Angelidakis et~al.(2017)Angelidakis, Makarychev, and
  Makarychev]{Angelidakis17}
H.~Angelidakis, K.~Makarychev, and Y.~Makarychev.
\newblock Algorithms for stable and perturbation-resilient problems.
\newblock In \emph{Proceedings of the Symposium on Theory of Computing}, pages
  438--451, 2017.

\bibitem[Backurs et~al.(2019)Backurs, Indyk, Onak, Schieber, Vakilian, and
  Wagner]{backurs2019scalable}
A.~Backurs, P.~Indyk, K.~Onak, B.~Schieber, A.~Vakilian, and T.~Wagner.
\newblock Scalable fair clustering.
\newblock In \emph{Proceedings of the International Conference on Machine
  Learning}, pages 405--413, 2019.

\bibitem[Bandyapadhyay et~al.(2021)Bandyapadhyay, Fomin, and
  Simonov]{bandyapadhyay2020coresets}
S.~Bandyapadhyay, F.~V. Fomin, and K.~Simonov.
\newblock On coresets for fair clustering in metric and {E}uclidean spaces and
  their applications.
\newblock In \emph{Proceedings of the International Colloquium on Automata,
  Languages, and Programming}, pages 23:1--23:15, 2021.

\bibitem[Bartal(1996)]{bartal1996probabilistic}
Y.~Bartal.
\newblock Probabilistic approximation of metric spaces and its algorithmic
  applications.
\newblock In \emph{Proceedings of the Conference on Foundations of Computer
  Science}, pages 184--193, 1996.

\bibitem[Bartal(1998)]{bartal1998approximating}
Y.~Bartal.
\newblock On approximating arbitrary metrices by tree metrics.
\newblock In \emph{Proceedings of the Symposium on Theory of Computing}, pages
  161--168, 1998.

\bibitem[Bera et~al.(2019)Bera, Chakrabarty, Flores, and
  Negahbani]{bera2019fair}
S.~Bera, D.~Chakrabarty, N.~Flores, and M.~Negahbani.
\newblock Fair algorithms for clustering.
\newblock In \emph{Advances in Neural Information Processing Systems}, pages
  4955--4966, 2019.

\bibitem[Bercea et~al.(2019)Bercea, Gro{\ss}, Khuller, Kumar, R{\"o}sner,
  Schmidt, and Schmidt]{bercea2019cost}
I.~O. Bercea, M.~Gro{\ss}, S.~Khuller, A.~Kumar, C.~R{\"o}sner, D.~R. Schmidt,
  and M.~Schmidt.
\newblock On the cost of essentially fair clusterings.
\newblock In \emph{Approximation, Randomization, and Combinatorial
  Optimization. Algorithms and Techniques}, 2019.

\bibitem[B{\"o}hm et~al.(2020)B{\"o}hm, Fazzone, Leonardi, and
  Schwiegelshohn]{bohm2020fair}
M.~B{\"o}hm, A.~Fazzone, S.~Leonardi, and C.~Schwiegelshohn.
\newblock Fair clustering with multiple colors.
\newblock \emph{arXiv preprint arXiv:2002.07892}, 2020.

\bibitem[Brubach et~al.(2020)Brubach, Chakrabarti, Dickerson, Khuller,
  Srinivasan, and Tsepenekas]{brubach2020pairwise}
B.~Brubach, D.~Chakrabarti, J.~Dickerson, S.~Khuller, A.~Srinivasan, and
  L.~Tsepenekas.
\newblock A pairwise fair and community-preserving approach to $k$-center
  clustering.
\newblock In \emph{Proceedings of the International Conference on Machine
  Learning}, pages 1178--1189, 2020.

\bibitem[Chakrabarty and Negahbani(2021)]{chakrabarty2021better}
D.~Chakrabarty and M.~Negahbani.
\newblock Better algorithms for individually fair $ k $-clustering.
\newblock \emph{arXiv preprint arXiv:2106.12150}, 2021.

\bibitem[Charikar et~al.(2002)Charikar, Guha, Tardos, and
  Shmoys]{charikar2002constant}
M.~Charikar, S.~Guha, {\'E}.~Tardos, and D.~B. Shmoys.
\newblock A constant-factor approximation algorithm for the $k$-median problem.
\newblock \emph{Journal of Computer and System Sciences}, 65\penalty0
  (1):\penalty0 129--149, 2002.

\bibitem[Chen et~al.(2019)Chen, Fain, Lyu, and
  Munagala]{chen2019proportionally}
X.~Chen, B.~Fain, L.~Lyu, and K.~Munagala.
\newblock Proportionally fair clustering.
\newblock In \emph{International Conference on Machine Learning}, pages
  1032--1041. PMLR, 2019.

\bibitem[Chierichetti et~al.(2017)Chierichetti, Kumar, Lattanzi, and
  Vassilvitskii]{chierichetti2017fair}
F.~Chierichetti, R.~Kumar, S.~Lattanzi, and S.~Vassilvitskii.
\newblock Fair clustering through fairlets.
\newblock In \emph{Advances in Neural Information Processing Systems}, pages
  5036--5044, 2017.

\bibitem[Chlamt{\'a}{\v{c}} et~al.(2022)Chlamt{\'a}{\v{c}}, Makarychev, and
  Vakilian]{chlamtavc2022approximating}
E.~Chlamt{\'a}{\v{c}}, Y.~Makarychev, and A.~Vakilian.
\newblock Approximating fair clustering with cascaded norm objectives.
\newblock In \emph{Proceedings of the 2022 Annual ACM-SIAM Symposium on
  Discrete Algorithms (SODA)}, pages 2664--2683, 2022.

\bibitem[Chouldechova(2017)]{chouldechova2017fair}
A.~Chouldechova.
\newblock Fair prediction with disparate impact: A study of bias in recidivism
  prediction instruments.
\newblock \emph{Big data}, 5\penalty0 (2):\penalty0 153--163, 2017.

\bibitem[Chouldechova and Roth(2020)]{chouldechova2018frontiers}
A.~Chouldechova and A.~Roth.
\newblock A snapshot of the frontiers of fairness in machine learning.
\newblock \emph{Communications of the ACM}, 63\penalty0 (5):\penalty0 82--89,
  2020.

\bibitem[Esmaeili et~al.(2020)Esmaeili, Brubach, Tsepenekas, and
  Dickerson]{esmaeili2020probabilistic}
S.~Esmaeili, B.~Brubach, L.~Tsepenekas, and J.~Dickerson.
\newblock Probabilistic fair clustering.
\newblock \emph{Advances in Neural Information Processing Systems}, 33, 2020.

\bibitem[Esmaeili et~al.(2021)Esmaeili, Brubach, Srinivasan, and
  Dickerson]{esmaeili2021fair}
S.~A. Esmaeili, B.~Brubach, A.~Srinivasan, and J.~P. Dickerson.
\newblock Fair clustering under a bounded cost.
\newblock \emph{Advances in Neural Information Processing Systems}, 2021.

\bibitem[Fakcharoenphol et~al.(2004)Fakcharoenphol, Rao, and
  Talwar]{fakcharoenphol2004tight}
J.~Fakcharoenphol, S.~Rao, and K.~Talwar.
\newblock A tight bound on approximating arbitrary metrics by tree metrics.
\newblock \emph{Journal of Computer and System Sciences}, 69\penalty0
  (3):\penalty0 485--497, 2004.

\bibitem[Feldman et~al.(2015)Feldman, Friedler, Moeller, Scheidegger, and
  Venkatasubramanian]{feldman2015certifying}
M.~Feldman, S.~A. Friedler, J.~Moeller, C.~Scheidegger, and
  S.~Venkatasubramanian.
\newblock Certifying and removing disparate impact.
\newblock In \emph{Proceedings of the SIGKDD International Conference on
  Knowledge Discovery and Data Mining}, pages 259--268, 2015.

\bibitem[Ghadiri et~al.(2021)Ghadiri, Samadi, and Vempala]{ghadiri2020fair}
M.~Ghadiri, S.~Samadi, and S.~Vempala.
\newblock Socially fair $k$-means clustering.
\newblock In \emph{Proceedings of the 2021 ACM Conference on Fairness,
  Accountability, and Transparency}, pages 438--448, 2021.

\bibitem[Huang et~al.(2019)Huang, Jiang, and Vishnoi]{huang2019coresets}
L.~Huang, S.~Jiang, and N.~Vishnoi.
\newblock Coresets for clustering with fairness constraints.
\newblock In \emph{Proceedings of the Conference on Neural Information
  Processing Systems}, 2019.

\bibitem[Jones et~al.(2020)Jones, Nguyen, and Nguyen]{jones2020fair}
M.~Jones, H.~Nguyen, and T.~Nguyen.
\newblock Fair $k$-centers via maximum matching.
\newblock In \emph{Proceedings of the International Conference on Machine
  Learning}, pages 4940--4949, 2020.

\bibitem[Jung et~al.(2020)Jung, Kannan, and Lutz]{jung2019center}
C.~Jung, S.~Kannan, and N.~Lutz.
\newblock A center in your neighborhood: Fairness in facility location.
\newblock In \emph{Proceedings of the Symposium on Foundations of Responsible
  Computing}, page 5:1–5:15, 2020.

\bibitem[Kariv and Hakimi(1979)]{KarivHakimi}
O.~Kariv and S.~L. Hakimi.
\newblock An algorithmic approach to network location problems. ii: The
  $p$-medians.
\newblock \emph{SIAM Journal on Applied Mathematics}, 37\penalty0 (3):\penalty0
  539--560, 1979.

\bibitem[Kearns and Roth(2019)]{kearns2019ethical}
M.~Kearns and A.~Roth.
\newblock \emph{The ethical algorithm: The science of socially aware algorithm
  design}.
\newblock Oxford University Press, 2019.

\bibitem[Kleinberg et~al.(2017)Kleinberg, Mullainathan, and
  Raghavan]{kleinberg2017inherent}
J.~Kleinberg, S.~Mullainathan, and M.~Raghavan.
\newblock Inherent trade-offs in the fair determination of risk scores.
\newblock In \emph{Proceedings of the Innovations in Theoretical Computer
  Science}, 2017.

\bibitem[Kleinberg et~al.(2018)Kleinberg, Lakkaraju, Leskovec, Ludwig, and
  Mullainathan]{kleinberg2018human}
J.~Kleinberg, H.~Lakkaraju, J.~Leskovec, J.~Ludwig, and S.~Mullainathan.
\newblock Human decisions and machine predictions.
\newblock \emph{The quarterly journal of economics}, 133\penalty0 (1):\penalty0
  237--293, 2018.

\bibitem[Kleindessner et~al.(2019)Kleindessner, Awasthi, and
  Morgenstern]{kleindessner2019fair}
M.~Kleindessner, P.~Awasthi, and J.~Morgenstern.
\newblock Fair {$k$}-center clustering for data summarization.
\newblock In \emph{Proceedings of the International Conference on Machine
  Learning}, pages 3448--3457, 2019.

\bibitem[Kleindessner et~al.(2020)Kleindessner, Awasthi, and
  Morgenstern]{kleindessner2020notion}
M.~Kleindessner, P.~Awasthi, and J.~Morgenstern.
\newblock A notion of individual fairness for clustering.
\newblock \emph{arXiv preprint arXiv:2006.04960}, 2020.

\bibitem[Li and Svensson(2016)]{li2016approximating}
S.~Li and O.~Svensson.
\newblock Approximating {$k$}-median via pseudo-approximation.
\newblock \emph{SIAM Journal on Computing}, 45\penalty0 (2):\penalty0 530--547,
  2016.

\bibitem[Mahabadi and Vakilian(2020)]{mahabadi2020individual}
S.~Mahabadi and A.~Vakilian.
\newblock Individual fairness for $k$-clustering.
\newblock In \emph{Proceedings of the International Conference on Machine
  Learning}, pages 6586--6596, 2020.

\bibitem[Makarychev and Vakilian(2021)]{makarychev2021approximation}
Y.~Makarychev and A.~Vakilian.
\newblock Approximation algorithms for socially fair clustering.
\newblock In \emph{Proceedings of the Conference on Learning Theory}, pages
  3246--3264. {PMLR}, 2021.

\bibitem[Marcinkowski et~al.(2020)Marcinkowski, Kieslich, Starke, and
  L{\"u}nich]{marcinkowski2020implications}
F.~Marcinkowski, K.~Kieslich, C.~Starke, and M.~L{\"u}nich.
\newblock Implications of ai (un-) fairness in higher education admissions: the
  effects of perceived ai (un-) fairness on exit, voice and organizational
  reputation.
\newblock In \emph{Proceedings of the Conference on Fairness, Accountability,
  and Transparency}, pages 122--130, 2020.

\bibitem[Micha and Shah(2020)]{micha2020proportionally}
E.~Micha and N.~Shah.
\newblock Proportionally fair clustering revisited.
\newblock In \emph{47th International Colloquium on Automata, Languages, and
  Programming (ICALP 2020)}. Schloss Dagstuhl-Leibniz-Zentrum f{\"u}r
  Informatik, 2020.

\bibitem[R{\"{o}}sner and Schmidt(2018)]{rosner18privacy}
C.~R{\"{o}}sner and M.~Schmidt.
\newblock Privacy preserving clustering with constraints.
\newblock In \emph{Proceedings of the International Colloquium on Automata,
  Languages, and Programming}, pages 96:1--96:14, 2018.

\bibitem[Schmidt et~al.(2019)Schmidt, Schwiegelshohn, and
  Sohler]{schmidt2019fair}
M.~Schmidt, C.~Schwiegelshohn, and C.~Sohler.
\newblock Fair coresets and streaming algorithms for fair {$k$}-means.
\newblock In \emph{Proceedings of the International Workshop on Approximation
  and Online Algorithms}, pages 232--251, 2019.

\bibitem[Tamir(1996)]{Tamir96}
A.~Tamir.
\newblock An $o(pn^2)$ algorithm for the $p$-median and related problems on
  tree graphs.
\newblock \emph{Operations Research Letters}, 19\penalty0 (2):\penalty0 59--64,
  1996.

\bibitem[Vakilian and Yal{\c{c}}{\i}ner(2021)]{vakilian2021improved}
A.~Vakilian and M.~Yal{\c{c}}{\i}ner.
\newblock Improved approximation algorithms for individually fair clustering.
\newblock \emph{arXiv preprint arXiv:2106.14043}, 2021.

\end{thebibliography}

\end{document}